\documentclass[10pt,twoside,a4paper]{article}

\usepackage{a4wide}
\usepackage{amssymb,amsmath,amsthm}
\usepackage{comment}
\usepackage{complexity}

\usepackage{enumitem}
\setlist{noitemsep}

\usepackage[algosection,algoruled,vlined,nofillcomment]{algorithm2e}
\DontPrintSemicolon

\usepackage{tikz}
\usetikzlibrary{arrows,decorations,shapes,shadows,positioning,plotmarks,patterns,calc,automata}

\usepackage{multirow}

\newcommand{\set}[1]{\ensuremath{\left\{#1 \right\}}}

\newcommand{\ceil}[1]{\ensuremath{\left\lceil #1 \right\rceil}}
\newcommand{\floor}[1]{\ensuremath{\left\lfloor #1 \right\rfloor}}

\renewcommand{\Re}{\mathbb{R}}
\newcommand{\Z}{\mathbb{Z}}

\newtheorem{theorem}{Theorem}[section]

\newtheorem{lemma}[theorem]{Lemma}

\tikzstyle{every path}=[>=stealth]

\tikzstyle{pool}=[cloud, draw,drop shadow,fill=white, aspect=2,minimum width=2em,minimum height=1.5em,font=\tiny]
\tikzstyle{device}=[draw,drop shadow,rounded corners,fill=white, aspect=2,minimum width=2em,minimum height=1.5em,font=\tiny]
\tikzstyle{arrow}=[font=\tiny]
\tikzstyle{desc}=[right=2em,font=\tiny]

\newcommand{\external}[1] {
	\draw[->] (#1.south)++(0.1em,0em) -- +(1em,-1em);
	\draw[->] (#1.south)++(0.2em,0em) -- +(1.5em,-1em);
	\draw[->] (#1.south)++(0.3em,0em) -- +(2em,-1em);
}

\title{Approximation algorithms for scheduling \\a group of heat pumps}
\author{Ji\v{r}\'i Fink \footnote{Supported by the Czech Science Foundation grant GA17-10090Y and by Center for Foundations of Modern Computer Science (Charles Univ. project UNCE/SCI/004). } \\ \small Department of Theoretical Computer Science and Mathematical Logic\\\small Faculty of Mathematics and Physics\\\small Charles University in Prague}
\date{}

\begin{document}
\maketitle
\begin{abstract}
This paper studies planning problems for a group of heating systems which supply the hot water demand for domestic use in houses. These systems (e.g. gas or electric boilers, heat pumps or microCHPs) use an external energy source to heat up water and store this hot water for supplying the domestic demands. The latter allows to some extent a decoupling of the heat production from the heat demand. We focus on the situation where each heating system has its own demand and buffer and the supply of the heating systems is coming from a common source. In practice, the common source may lead to a coupling of the planning for the group of heating systems. The bottleneck to supply the energy may be the capacity of the distribution system (e.g. the electricity networks or the gas network). As this has to be dimensioned for the maximal consumption, it is important to minimize the maximal peak. This planning problem is known to be \NP-hard. We present polynomial-time approximation algorithms for four variants of peak minimization problems, and we determine the worst-case approximation error.
\end{abstract}

\section{Introduction}

In modern society, a significant amount of energy is consumed for heating water \cite{heating_water}. Almost every building is connected to a district heating system or equipped with appliances for heating water locally. Typical appliances for heating water are electrical and gas heating systems, heat pumps and Combined Heat and Power units (microCHP). The heated water is stored in buffers to be prepared for the demands of the inhabitants of the building.

In this paper we consider a local heating system which consist of
\begin{itemize} \itemsep=0em
	\item a supply which represents some source of energy (electricity, gas),
	\item a converter which converts the energy into heat (hot water),
	\item a buffer which stores the heat for later usage and
	\item a demand which represents the (predicted) consumption profile of heat.
\end{itemize}
A more formal definition of the considered setting for local heating and the used parameters and variables is given in Section \ref{sec:overview}. The presented model can consider arbitrary types of energy but in this paper we use \textit{electricity} and \textit{heat} to distinguish consumed and produced energy. However, this simple model of a local heating system can not only be applied for heating water but has many other applications, e.g. heating demand of houses, fridges and freezers and inventory management.

The combination of a heating device and a buffer gives some freedom in deciding when the heat has to be produced. To use this freedom in a proper way, different objectives may be considered in practice. The energy used for heating is transported from a supply to the heating systems by electrical networks or gas pipes. These transport media have to be able to transport all the used energy and therefore have to be dimensioned for the maximal consumption peak of all houses connected to the transport network. Thus, minimizing the maximal consumption over all these houses may decrease investments in the distribution networks. This leads to planning problems for a group of heating systems which minimize peak where peak is the maximal consumption of electricity over the planning period.

In order to show that our algorithm can be adopted for various scenarios, we consider four variants of peak shaving problems. In the basic case, we assume that only heating systems are connected to electricity grid, so the goal is to find a scheduling of these heating systems which minimize the maximal consumption of electricity of all heating systems during a planing horizon. In the second case, we consider more realistic scenario where every house has some other devices consuming or producing electricity. Since we may not be allowed to control these devices, we assume that the total electrical consumption (called base load) is given. This base load is added to the consumption of all heating systems, so the objective in this case is minimizing the peak of the sum. As the production of electricity in family houses is increasing during last decade due to government subsidies on photovoltaic panels (PV) and combine heat and power units (microCHP), it is necessary to take into account not only the overconsumption of electricity but also the overproduction in local districts. This leads us to the third case which minimizes the maximal absolute value of the total electricity consumption. However, if the consumption of electricity is significantly higher than the production, minimizing the maximal consumption may give the same result as minimizing the maximal absolute value. In this case, it may be more practical to minimize the fluctuation or bandwidth. In other words, the goal is minimizing the difference between the maximal and the minimal consumption during the planning horizon. Formal definitions of all these problems are given in Section \ref{sec:overview}.

In our previous study \cite{cost_peak} we proved that the basic case is \NP-hard problem and therefore all variants studied in this paper are also \NP-hard. The computational hardness of these problems strongly relates to the 3-partition problem (for definition, see e.g. \cite{Garey}). One possibility to avoid the hardness of 3-partition problem is considering a special case of our heating problems. For example, paper \cite{cost_peak} also presents a polynomial-time algorithm for minimizing the maximal peak in the case where all converters consume the same amount of electricity when running. This algorithm reduces every instance of the special case of heating problem into a job scheduling problem $P_m|r_i,p_i=1,\text{chains}|L_{\max}$ which was proved by Dror et al.\cite{Dror98} to be polynomial. This reduction shows strong relation between our heating problems and job scheduling. Other possibility to avoid the computational hardness is developing approximation algorithms which is the task of this paper.

Although for small scale case studies, mix integer linear programming (MILP) solvers are able to find optimal solution in reasonable time \cite{milp_comfort,simple_control}, using MILP solvers become impractical in larger case studies due to high time and memory demands. Therefore, we develop polynomial-time approximation algorithms for four variants of peak shaving problems in this paper, and we determine the worst-case approximation error. Although, the classical measure of quality of approximation algorithms is the relative approximation error (see e.g. \cite{Cormen}), we use the absolute approximation error. This is due the fact that the optimal value may be zero, so the relative error may be undefined; see Section \ref{sec:error}.
\begin{figure*}[t]
\begin{center}

\resizebox{0.5\textwidth}{!}{
\begin{tikzpicture}

\node[device,minimum height=11em,minimum width=4.8em] at (-5em,-9.3em) {} node[desc] at (-9.5em,-14em) {heating system};
\node[device,minimum height=11em,minimum width=4.8em] at (0.2em,-9.3em) {} node[desc] at (-4.3em,-14em) {heating system};

\node[device] (market) at (0em,0em) {supply};
\node[device] (converter) at (0em,-5em) {converter} edge[<-] (market);
\node[desc] at (-1em,-2.5em) {Other heating systems};
\node[desc] at (1em,-9em) {\ldots};
\node[device] (buffer) at (0em,-8.5em) {buffer};
\node[device] (demand) at (0em,-12em) {demand};
\path[->] (converter) edge (buffer);
\path[->] (buffer) edge (demand);

\node[device] (c2) at (-5em,-5em) {converter} edge[<-] (market);
\node[device] (b2) at (-5em,-8.5em) {buffer} edge[<-] (c2);
\node[device] at (-5em,-12em) {demand} edge[<-] (b2);

\external{market};

\end{tikzpicture}
}

\caption{Schematic picture of heating systems split into converters, buffers and demands. A group of those heating systems is connected to a common supply of energy.}
\label{model}
\end{center}
\end{figure*}
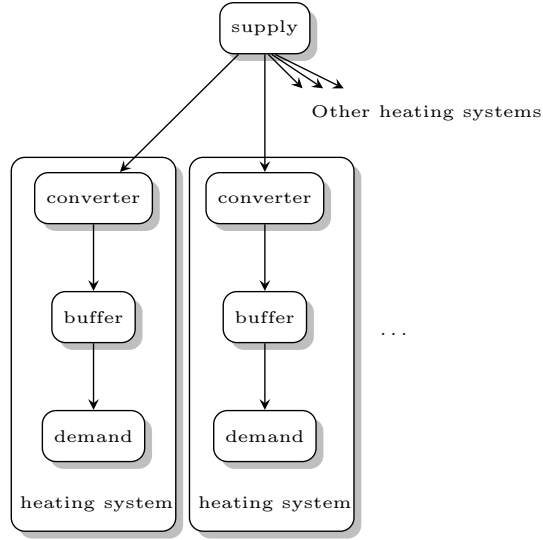

\subsection{Problem statement and results} \label{sec:overview}

This section presents a mathematical description of the studied model and a summary of the results of this paper. The used symbols, parameters and variables are summarized in Appendix.

First, we consider a discrete time model for the considered problem, meaning that we split the planning period into $T$ time intervals of the same length. We consider sets ${\cal C} = \set{1,\dots,C}$ of $C$ heating systems and ${\cal T} = \set{1,\dots,T}$ of $T$ time intervals. For mathematical purposes, we separate a heating system into a converter, a buffer and a demand; see Figure \ref{model}.

We consider a simple converter which has only two states: In every time interval the converter is either turned on or turned off. The amount of consumed electricity is $E_c$ and the amount of produced heat (or any other form of energy) is $H_c$ during one time interval in which the converter $c \in \cal{C}$ is turned on. We assume that $E_c$ is positive when a converter $c$ is consuming electricity (the converter is e.g. electrical boiler, heat pump), and negative when the converter is producing electricity (e.g. microCHP). As we discuss later, we assume that $E_c$ is non-zero for all converters $c \in \cal C$.  Let $x_{c,t} \in \set{0,1}$ be the variable indicating whether the converter $c \in \cal{C}$ is running in time interval $t \in \cal{T}$.

The state of charge of a buffer $c \in \cal{C}$ at the beginning of time interval $t \in \cal{T}$ is denoted by $s_{c,t}$ which represents the amount of heat in the buffer. Note that $s_{c,T+1}$ is the state of charge at the end of planning period. The state of charge $s_{c,t}$ is limited by a lower bound $L_{c,t}$ and an upper bound $U_{c,t}$. Those two bounds are usually constant over time: the upper bound $U_{c,t}$ is the capacity of buffer and the lower bound $L_{c,t}$ is mostly zero. However, it may be useful to allow different values, e.g. in this paper, we assume that $L_{c,1} = U_{c,1}$, so the initial state of charge $s_{c,1}$ is fixed.

The amount of consumed heat by the inhabitants of the house from heating system $c \in \cal{C}$ during time interval $t \in \cal{T}$ is denoted by $D_{c,t}$. This amount is assumed to be given and is called the demand of heating system $c$. The total amount of electricity consumed by other devices in all houses $\cal C$ during time interval $t$ is called base load and it is denoted by $F_t$. In this paper, we study off-line problems, so we assume that demands $D_{c,t}$ and the base load $F_t$ are given for the whole planning period.

The operational variables of the converters $x_{c,t}$ and the states of charge of buffers $s_{c,t}$ are restricted by the following invariants.
\begin{eqnarray}
	s_{c,t+1} = s_{c,t} + H_c x_{c,t} - D_{c,t} & \text{ for } & c \in {\cal C},\, t \in {\cal T} \label{eq:charging} \\
	L_{c,t} \le s_{c,t} \le U_{c,t} & \text{ for } & c \in {\cal C},\, t \in {\cal T} \cup \set{T+1} \label{eq:buffer} \\
	x_{c,t} \in \set{0,1} & \text{ for } & c \in {\cal C},\, t \in {\cal T} \label{eq:converter}
\end{eqnarray}

The basic objective function of this paper is minimizing the peak electricity consumption.
\begin{eqnarray}
\text{Basic peak shaving:} &&  \text{minimize } m \label{eq:peak} \\
								&& \text{where }  m \ge \sum_{c \in \cal C} E_c x_{c,t} \text{ for } t \in \cal T \label{eq:m}
\end{eqnarray}

In more general case, we also add the base load $F_t$ to the electricity consumed by all converters.
\begin{eqnarray}
\text{Maximal peak shaving:} &&  \text{minimize } m \nonumber\\
&& \text{where }  m \ge F_t + \sum_{c \in \cal C} E_c x_{c,t} \text{ for } t \in \cal T \label{eq:maxpeak}
\end{eqnarray}
Since the basic peak shaving problem is a special case of the maximal peak shaving problem with the base load where $F_t = 0$, it suffices to develop approximation algorithm for the maximal peak shaving problem.

In districts with large PV installations, it is common that more electricity is produced than consumed during summer. In this case, the distribution grid needs to be dimensioned not only for the maximal consumption but also for the maximal production. One possible approach to capture this issue is introduction an objective function which minimize the maximal absolute value of the total electricity consumption.
\begin{eqnarray}
\text{Absolute peak shaving:} &&  \text{minimize } m \nonumber\\
&& \text{where }  m \ge \left| F_t + \sum_{c \in \cal C} E_c x_{c,t} \right| \text{ for } t \in \cal T \label{eq:maxabs}
\end{eqnarray}
Note that the absolute and the maximal peak shaving problems give the same result when $E_c$ and $F_t$ are positive. Similarly, when the average of the total electrical consumption over planning horizon is sufficiently far from zero, only one bound of \eqref{eq:maxabs} dominates the solution (the upper one for overconsumption and the lower one for overproduction) and the other bound may have negligible influence. In such cases, an external source of energy is needed to balance the difference between the production and the consumption. The external source of energy may be a conventional generator of constant production which may be unable to balance large fluctuations. One possible approach to reduce such fluctuation is maximizing the minimal total electricity consumption together with minimizing the maximal total consumption. In order to incorporate these two goals into one objective function, we introduce the fluctuation peak shaving problem which minimize the difference between the maximal and the minimal consumption during whole planning horizon.
\begin{eqnarray}
\text{Fluctuation peak shaving:} &&  \text{minimize } m_u - m_l \nonumber\\
&& \text{where }  m_l \le F_t + \sum_{c \in \cal C} E_c x_{c,t} \le m_u \text{ for } t \in \cal T \label{eq:maxfluctuation}
\end{eqnarray}
In the last objective function, $m_l$ and $m_u$ are the lower and upper bounds on the electricity consumption, respectively.


As we discuss above, we assume that $E_c$ is non-zero for all converters $c \in \cal C$ since converters $c$ with electricity consumption $E_c = 0$ have no influence on any objective function studied in this paper, and therefore they can be scheduled independently on other converters.

\subsection{Approximation error} \label{sec:error}

In this paper, we develop polynomial-time approximation algorithms for four variants of peak shaving problems introduced above. In order to determine the approximation error, we mainly consider the absolute error instead of the relative error usually used in literature (see e.g. \cite{Cormen}). Formally, let $m^O$ be the optimal value of objective function of the maximal peak shaving problem and $m^A$ be the value of objective function of a solution found by our algorithm. The relative error is the maximal ratio $\frac{m^A}{m^O}$ over all instances of the problem, and the absolute error is the  maximal difference $m^A - m^O$. We prove that our algorithm always finds a solution satisfying $m^A - m^O \le E$ where $E = \max_{c \in \cal C} |E_c|$. We obtain the same approximation error for the absolute peak shaving problem meaning that the interval $\langle -m^O,m^O \rangle$ contains total electrical consumption $F_t + \sum_{c \in C}E_c x_{c,t}$ for all $t \in \cal T$ of the optimal solution needs to be extended into $\langle -m^A,m^A \rangle \subseteq \langle -m^O-E,m^O+E \rangle$ for the approximated solution found by our algorithm. Similarly for the fluctuation peak shaving problem, the interval of total electricity consumptions $\langle m^O_l,m^O_u \rangle$ for the optimal solution needs to be extended into $\langle m^A_l,m^A_u \rangle \subseteq \langle m^O_l-E, m^O_u+E \rangle$ for the approximation solution. This implies that the absolute error for the fluctuation peak shaving problem is at most $2E$.

Reader may ask why we consider the absolute error instead of the relative error. Observe that the optimal value $m^O$ may be zero in all four objective function since $E_c$ may be positive for some $c$ and negative for other $c$ and the production and consumption of electricity may be perfectly balanced. This implies that the relative error may be undefined. Even if we restrict to positive values of $E_c$, the base load $F_t$ may cause the optimal value to be zero. Even if we further set $F_t$ to be always zero, then the fluctuation peak shaving problem may have zero optimal value. Furthermore, the decision problem whether the fluctuation peak shaving problem have zero optimal value is strongly \NP-complete by a simple reduction from the 3-partition problem (see e.g. \cite{cost_peak}). Therefore, the classical relative approximation error have sense only for the basic, maximal and absolute peak shaving problems when $E_c > 0$ and $F_t = 0$ and in this case the objectives gives the same value. In the case, without loss of generality we can assume that every converter has to be run at least once (i.e. $L_{c,T+1} > U_{c,1}$). Since $m^O \ge E$ and $m^A \le m^O + E$, the relative approximation error is at most 2. However, since our problems are strongly \NP-complete, there is no fully polynomial-time approximation scheme unless $\P = \NP$. It is question whether there exists a polynomial-time approximation scheme for this case.

\subsection{Motivation}

The considered problems originate from a project called \textit{MeppelEnergie} which plans to build a group of houses and a biogas station in Meppel, a small city in the Netherlands\footnote{For more details, see websites \url{http://www.utwente.nl/ctit/energy/projects/meppel.html} and \url{http://www.meppelwoont.nl/nieuwveense-landen/}}. In this project, the houses will have a heat pump for space heating and tap demands. Due to Dutch legislation, the biogas station will provide electricity only to those heat pumps. Therefore, the heat pumps should be scheduled in such a way that they only consume, if possible, the electricity produced by the biogas station. If this is not possible, the remaining energy has to be bought on the electricity market at minimal cost.

The study \cite{simple_control} shows that some central control of all heat pumps is necessary to avoid large peak loads. Therefore, our task is to design one or more algorithms to control all heat pumps. The first of our proposed algorithms is called \textit{global MILP control} which uses a Mix Integers Linear Programming solver to find an optimal (or near to optimal) solution of the minimizing peak problem. The paper \cite{simple_control} shows that this approach can be used only for small number of houses. For larger number of houses, a faster algorithm for the minimizing peak problem is necessary but the problem is \NP-hard. Therefore, we try to find an easier problem which can be solved faster (polynomial algorithm for the case where all heat pump consume the same amount of energy; FPT algorithm for the general case \cite{simple_control}). In practice, it may be sufficient to find a solution which is close to the optimum. One such approximation algorithm is presented in \cite{simple_control} but no worst-case analysis is given. Developing an approximation algorithm with proven worst-case analysis is a task of this paper.


\subsection{Overview}

Now, we present basic ideas of our approximation algorithms and organization of the paper. Section \ref{sec:application} presents related works and more details on applications of the results of the paper. Section \ref{sec:preprocessing} reformulates the studied problem into a simpler form. The basic idea of our approximation algorithms is finding an optimal solution $y$ of the relaxed problem and rounding this relaxed solution into a binary solution $x$. Section \ref{sec:idea} gives more details about this idea. Section \ref{sec:forest} studies the structure of non-integer values in the relaxed solution $y$. Section \ref{sec:order} constructs an order of these non-integer values. The relaxed solution $y$ is rounded in this order into an integer solution $x$ in Section \ref{sec:rounding}. Section \ref{sec:conclusion} concludes this paper with remarks and open problem. Parameters, variables and symbols are listed in Appendix.

\section{Related works and applications} \label{sec:application}

In the following we present related literature and give some possible applications of this model.

Some related works can be found in the inventory management and lot-sizing literature (see e.g. \cite{drexl1997lot,karimi2003capacitated} for reviews). In inventory control problems (see e.g. \cite{axsater2006inventory}) a buffer may represent an inventory of items, whereby a converter represents the production of items and demand represents the order quantities. As our problem consists of only one commodity, the single item lot sizing problem is related (see \cite{brahimi2006single} for a review). Wagner and Whitin \cite{wagner1958dynamic} presented an $\mathcal{O}(T^2)$ algorithm for the uncapacitated lot-sizing problem which was improved by Federgruen and Tzur \cite{federgruen1991simple} to $\mathcal{O}(T \log T)$. On the other hand, Florian, Lenstra and Rinnooy \cite{florian1980deterministic} proved that the lot-sizing problem with upper bounds on production and order quantities is NP-complete. Computational complexity of the capacitated lot sizing problems is studied in \cite{bitran1982computational}. Pessoa at.al. \cite{de2019automatic} studied multiple variants of Multi-level capacitated lot-sizing problem which is an NP-hard problem, so they presented an automatic algorithm-generation approach based on heuristics and a multi-population genetic algorithm. Quezada et.al. \cite{quezada2019stochastic} proposed a stochastic dual dynamic integer programming algorithm for the multi-echelon multi-item lot-sizing problem. There are also papers about approximation schemes for single-item lot-sizing problems with economic objectives (see e.g. \cite{van2001fully,chubanov2006fptas,chubanov2012fptas}) and multi-item variants (\cite{levi2008approximation}).
Absi and Kedad-Sidhoum \cite{absi2008multi} presented a branch-and-cut algorithm for the multi-item capacitated lot-sizing problem. The main difference in this paper is the objective function which consider minimizing of fluctuations of electricity consumption instead of overall production and holding cost.

One other related area is vehicle routing and scheduling (see e.g. \cite{laporte1992vehicle} for an overview of this area). For example, Lin, Gertsch and Russell \cite{hong2007linear} studied optimal vehicle refuelling policies. In their model, a refuelling station can provide an arbitrary amount of gas while our converter is restricted to two possible states of heat generation. Other papers on vehicle refuelling policies are more distant from our research since they consider that a car is routed on a graph (see e.g. \cite{sweda2012finding,lin2008finding}).

Balancing electricity using a group of microCHPs installed in individual households was studied by Bosman et al.\cite{bosman2012planning}. On one hand, they considered more complex model of a converter (e.g. minimal running time and starting profile) which makes the model more accurate in practice. On the other hand, they presented only experimental results without any study of worst-case behaviour of their algorithms.

Peak load shaving is a classical problem in smart grids and demand side management studied in many papers (see e.g. \cite{rahimi2013simple,wang2013grid,alam2014controllable,zhao2013strategies}). This paper considers peak shaving as an objective, while some other studies set the minimal or the maximal electrical consumption as a hard constraint and optimize different objectives, for example \cite{bosman2012planning} maximizes the profit on the electricity market and \cite{Klauw2014Scheduling} minimizes the number of charging cycles of batteries.

\section{Reformulation of the problem} \label{sec:preprocessing}

In this section, we simplify the problem presented in Section \ref{sec:overview}. Since we used the same reformulation in our previous papers \cite{cost_peak,greedy}, we present only the basic idea of the reformulation here for sake of completeness. The goal of this reformulation is to replace conditions \eqref{eq:charging} and \eqref{eq:buffer} by one condition \eqref{eq:sum}.

First, we expand the recurrence formula \eqref{eq:charging} into an explicit equation.
$$s_{c,t+1} = s_{c,1} + \sum_{i=1}^t H_c x_{c,i} - \sum_{i=1}^t D_{c,i}$$
Since we assume that the initial state of charge satisfies $s_{c,1} = L_{c,1} = U_{c,1}$, we can replace $s_{c,1}$ by $L_{c,1}$ and substitute into inequalities \eqref{eq:buffer} and after elementary operations we obtain
$$\frac{L_{c,t+1} - L_{c,1} + \sum_{i=1}^t D_{c,i}}{H_c} \le \sum_{i=1}^t x_{c,i} \le \frac{U_{c,t+1} - L_{c,1} + \sum_{i=1}^t D_{c,i}}{H_c}.$$
Here, the sum $\sum_{i=1}^t x_{c,i}$ is restricted by a lower and an upper bounds which depend only on input parameters, so these bounds can be easily precomputed to obtain
\begin{eqnarray}
A_{c,t} \le \sum_{i=1}^t x_{c,i} \le B_{c,t} \text{ for } c \in {\cal C}, t \in {\cal T}. \label{eq:sum}
\end{eqnarray}
Further analysis presented in \cite{cost_peak,greedy} proves that bounds $A_{c,t}$ and $B_{c,t}$ are integral. Furthermore, we can observe that a binary solution $x$ satisfies \eqref{eq:charging} and \eqref{eq:buffer} if and only if $x$ satisfies \eqref{eq:sum} where the state of charge $s$ is directly computed from $x$ by \eqref{eq:charging}.

\section{Ideas of our algorithms}\label{sec:idea}

In this section, we present ideas of our algorithms. The original problems minimize objectives \eqref{eq:maxpeak}, \eqref{eq:maxabs} and \eqref{eq:maxfluctuation} subject to \eqref{eq:charging}, \eqref{eq:buffer} and \eqref{eq:converter}. The previous section shows how conditions \eqref{eq:charging}, \eqref{eq:buffer} can be simplified to the condition \eqref{eq:sum} in linear time. The main contribution of the paper is presenting an approximation polynomial time algorithm minimizing objectives \eqref{eq:maxpeak}, \eqref{eq:maxabs} and \eqref{eq:maxfluctuation} subject to \eqref{eq:converter} and \eqref{eq:sum}.

In the next step, we consider the relaxed problem obtained by replacing integer constrains \eqref{eq:converter} by inequalities
\begin{equation}\label{eq:relaxed}
	0 \le x_{c,t} \le 1  \text{ for every } c \in {\cal C} \text{ and } t \in {\cal T}.
\end{equation}
Problems minimizing objectives \eqref{eq:maxpeak}, \eqref{eq:maxabs} and \eqref{eq:maxfluctuation} subject to \eqref{eq:sum} and \eqref{eq:relaxed} are (non-integral) Linear Programming problems, so an optimal relaxed solution $y$ can be found in polynomial time. If this relaxed problem has no feasible solution, then there is no integral solution. Hence, in the rest of this paper we assume that the relaxed problem is feasible and $y$ denotes an optimal relaxed solution.  Sections \ref{sec:forest}, \ref{sec:order} and \ref{sec:rounding} present how the optimal relaxed solution $y$ can be rounded to approximated integral solution $x$ satisfying conditions in the following theorem.

\begin{theorem}\label{thm:round}
For every $y \in \Re^{ {\cal C} \times {\cal T}}$ with $0 \le y_{c,t} \le 1$ for every $t \in \cal{T}$ and $c \in \cal C$ there exists $x \in \Re^{ {\cal C} \times {\cal T}}$ with $x_{c,t} \in \set{0,1}$ for every $t \in \cal{T}$ and $c \in \cal C$ such that
\begin{eqnarray}
\floor{\sum_{i=1}^t y_{c,i}} \le \sum_{i=1}^t x_{c,i} \le \ceil{\sum_{i=1}^t y_{c,i}} && \text{ for every } c \in {\cal C} \text{ and } t \in \cal T \label{eq:bounds} \\
\left| \sum_{c \in \cal C} E_c x_{c,t} - \sum_{c \in \cal C} E_c y_{c,t} \right| \le \max_{c \in \cal C} |E_c| && \text{ for every } t \in \cal T \label{eq:approx}
\end{eqnarray}
and $x$ can be found in time $\mathcal{O}(T^2C^2)$.
\end{theorem}

Since we need to repeat sums used in constrains \eqref{eq:m} and \eqref{eq:sum} quite often, we define the following two symbols to simplify the notation.
$$\hat{x}_{c,t} = \sum_{i=1}^t x_{c,i} \text{ \qquad and \qquad } \breve{x}_t = F_t + \sum_{c \in \cal C} E_c x_{c,t}$$
Symbols $\hat{y}_{c,t}$, $\breve{y}_t$, $\hat{z}_{c,t}$ and $\breve{z}_t$ are defined analogously for solutions $y$ and $z$.

In summary, the approximated integral solution $x$ can be found in polynomial time and the absolute approximation error is determined in the following theorem.
\begin{theorem}\label{thm:error}
There exists an appropriate polynomial time algorithms minimizing objectives \eqref{eq:maxpeak}, \eqref{eq:maxabs} and \eqref{eq:maxfluctuation} subject to \eqref{eq:charging}, \eqref{eq:buffer} and \eqref{eq:converter} with an absolute error $E$ for Maximal peak shaving \eqref{eq:maxpeak} and Absolute peak shaving \eqref{eq:maxabs} and $2E$ for Fluctuation peak shaving problem \eqref{eq:maxfluctuation} where $E = \max_{c \in \cal C} |E_c|$.
\end{theorem}
\begin{proof}
Bounds $A$ and $B$ and a relaxed optimal solution $y$ can be found in polynomial time. By Theorem \ref{thm:round}, we can find an approximated integral solution $x$ in polynomial time. The solution $x$ is feasible since assumptions \eqref{eq:bounds} implies $A_{c,t} \le \floor{\hat{y}_{c,t}} \le \hat{x}_{c,t} \le \ceil{\hat{y}_{c,t}} \le B_{c,t}$ for every $c \in {\cal C}$ and $t \in \cal T$ as required by \eqref{eq:sum}.

Let $m^O$ be the optimal value of an objective function and $m^A$ be the value of an objective function of the approximated solution $x$. For the Maximal peak shaving problem, it holds that $m^O \ge \max_t \breve{y}_t$ since $y$ is an optimal relaxed solution. From \eqref{eq:approx} it follows that $m^A = \max_t \breve{x}_t \le \max_t \breve{y}_t + E$ and
$$m^A - m^O \le \max_t \breve{y}_t + E - \max_t \breve{y}_t = E$$
and so the absolute approximation error is at most $E$. Similarly for the Absolute peak shaving problem, it holds that $m^O \ge \max_t |\breve{y}_t|$. Hence,
$$m^A - m^O \le \max_t |\breve{x}_t| - \max_t |\breve{y}_t| \le \max_t |\breve{y}_t| + E - \max_t |\breve{y}_t| = E$$
and so the absolute approximation error is also at most $E$. Finally for the Fluctuation peak shaving problem, it holds that $m^O \ge \max_t \breve{y}_t - \min_t \breve{y}_t$. Hence,
$$m^A - m^O \le (\max_t \breve{x}_t - \max_t \breve{y}_t) + (\min_t \breve{y}_t - \min_t \breve{x}_t) \le 2E$$
and so the absolute approximation error is at most $2E$.
\end{proof}

\section{Structure of non-integer values in relaxed solutions} \label{sec:forest}

In this section, we study properties of vertices of a polytope
$$P = \set{x : \text{ \eqref{eq:relaxed} and  \eqref{eq:bounds} hold, and $\breve{x}_t = \breve{y}_t$ holds for every $t \in \cal T$}}$$
where $y$ is an optimal relaxed solution. Clearly, $P$ contains only relaxed optimal solutions including $y$.

The structure of $P$ is described in Lemmas \ref{lem:vertex} and \ref{lem:forest}. To simplify the notation, we define
$$T(t_1,t_2) =
\begin{cases}
\set{t_1, t_1+1, \ldots, t_2-1} & \text{ if } t_1 \le t_2 \\
\set{t_2, t_2+1, \ldots, t_1-1} & \text{ otherwise}
\end{cases}$$
for $t_1,t_2 \in \cal T$. Let $T_c = \set{t \in {\cal T}: y_{c,t} \notin \Z}$ and $\sim_c$ be a relation on $T_c$ such that $t_1 \sim_c t_2$ if and only if $\hat{y}_{c,t} \notin \Z$ for every $t \in T(t_1, t_2)$. Since the relation $\sim_c$ is an equivalence on $T_c$, we denote by $S_c$ the set of equivalence classes for every converter $c$. Our goal is to reduce the number of non-integral values $\varphi(y) = |\set{(c,t);\; y_{c,t} \notin \Z}| + |\set{(c,t);\; \hat{y}_{c,t} \notin \Z}|$.

\begin{lemma}\label{lem:vertex}
If there exists a sequence of converters $c_1, \ldots, c_k$ and a sequence of distinct time intervals $t_1, \ldots, t_k, t_{k+1}=t_1$ for $k \ge 2$ such that conditions
\begin{enumerate}[label=(B\arabic*),ref=(B\arabic*)]
\item $y_{c_i,t_i}, y_{c_i,t_{i+1}} \notin \Z$ \label{itm:Z}
\item $t_i \sim_{c_i} t_{i+1}$ \label{itm:Ti}
\end{enumerate}
hold for every $i \in \set{1,\ldots,k}$, then $y$ is not a vertex of $P$. Furthermore, there exists a solution $z \in P$ such that $\varphi(z) < \varphi(y)$ which can be found in time $\mathcal{O}(TC)$.
\end{lemma}
\begin{proof}
We define a line segment of points $z(\alpha)$ parametrized by $\alpha$ as follows
\begin{itemize}
\item $z_{c_i,t_i}(\alpha) = y_{c_i,t_i} + \frac{\alpha}{E_{c_i}}$
\item $z_{c_i,t_{i+1}}(\alpha) = y_{c_i,t_{i+1}} - \frac{\alpha}{E_{c_i}}$
\end{itemize}
for every $i \in \set{1,\ldots,k}$ and $z_{c,t}(\alpha) = y_{c,t}$ otherwise. For every $\alpha$ the solution $z(\alpha)$ satisfies $\breve{z}_t(\alpha) = \breve{y}_t$. Furthermore, $\hat{z}_{c,t}(0) = \hat{y}_{c,t}$ and $z_{c,t}(\alpha)$ is linear in $\alpha$. Hence, for some sufficiently small $\epsilon > 0$ inequalities
\begin{center}\begin{tabular}[h]{ll}
$0 \le z_{c,t}(\alpha) \le 1$ &
\multirow{2}{*}{} \\
$\floor{\hat{y}_{c,t}} \le \hat{z}_{c,t}(\alpha) \le \ceil{\hat{y}_{c,t}}$
\end{tabular}\end{center}
hold for every $t \in \cal T$, $c \in \cal C$ and $-\epsilon \le \alpha \le \epsilon$. Hence, whole line segment of points $z(\alpha)$ for $\alpha \in \langle -\epsilon, \epsilon \rangle$ belongs into the polytope $P$ which proves that $y$ is not a vertex of $P$.

Consider the maximal $\alpha$ such that $z(\alpha) \in P$ and let $z = z(\alpha)$ which can be found in time $\mathcal{O}(TC)$. Observe that maximality of $\alpha$ implies $\varphi(z) < \varphi(y)$.
\end{proof}

Now, we define a bipartite graph $G(y)$ for a solution $y$ as follows. One partite of vertices consists of the set of time intervals $\cal T$ and the other partite is the set of pairs $(c,W)$ where $c \in \cal C$ and $W \in S_c$. Time interval $t$ and a pair $(c,W)$ are connected by an edge if $t \in W$. Note that the graph $G(y)$ has exactly one edge for every non-integral value in $y$. The number of vertices of $G(y)$ is at most $(C+1)T$ and the number of edges of $G(y)$ is $\sum_{c \in \cal C} T_c \le CT$.

\begin{lemma}\label{lem:forest}
If $G(y)$ is not a forest, then there are sequences of converters $c_1, \ldots, c_k$ and time intervals $t_1, \ldots, t_k$ satisfying assumptions of Lemma \ref{lem:vertex} which can be found in time $\mathcal{O}(TC)$.
\end{lemma}
\begin{proof}
For sake of a contradiction, let us assume that $G(y)$ has a cycle on vertices
$$t_1,(c_1,W_1),t_2,(c_2,W_2),\ldots,t_k,(c_k,W_k), t_{k+1}=t_1.$$
Since $t_i, t_{i+1} \in W_i \subseteq T_{c_i}$, conditions \ref{itm:Z} and \ref{itm:Ti} hold, so sequences $t_1,\ldots,t_k$ and $c_1,\ldots,c_k$ satisfy conditions of Lemma \ref{lem:vertex}. This contradicts the assumption that $y$ is a vertex of $P$. A cycle in a graph can be found in time linear in the size of the graph.
\end{proof}

Interested reader may observe that $G(y)$ is a forest if and only if $y$ is a vertex of $P$. However for our purposes, only implication in Lemma \ref{lem:forest} is important.

\begin{lemma}\label{lem:z}
We can find a solution $z \in P$ such the graph $G(z)$ is a forest in time $\mathcal{O}(T^2C^2)$.
\end{lemma}
\begin{proof}
First, we consider the optimal relaxed solution $y$ to be an initial solution $z$. Note that $\varphi(z) \le 2CT$. We repeatedly use Lemma \ref{lem:vertex} and Lemma \ref{lem:forest} until $G(z)$ is a forest. In every step the potential $\varphi(z)$ decreases by at least one and it terminates after at most $2CT$ steps since $G(z)$ is a forest if $\varphi(z) = 0$. Time complexity of one step is $\mathcal{O}(TC)$ so the complexity of finding $z \in P$ such the graph $G(z)$ is a forest is $\mathcal{O}(T^2C^2)$.
\end{proof}

\section{Rounding order}\label{sec:order}

In this section we consider a solution $z$ provided by Lemma \ref{lem:z} and we construct a sequence $(c_1,W_1),\ldots,(c_k,W_k)$ of all vertices of the second partite of $G(z)$ determining the order in which non-integer values of the solution $z$ are rounded. Next section rounds all non-integer values $z_{c_i,t}$ for $t \in W_i$ sequentially for $i = 1, \ldots, k$. The construction of the order is described in the following lemma. To simplify the notation, let $G$ be the graph $G(z)$.

\begin{lemma}\label{lem:order}
There is a sequence $(c_1,W_1),\ldots,(c_k,W_k)$ of all vertices of the second partite of $G$ such that for every $i$ vertex $(c_i,W_i)$ has at most one non-leaf neighbour $t_i$ in the graph
$$G_i = G \setminus \set{(c_{i+1},W_{i+1}), \ldots, (c_k,W_k)}.$$
The sequence can be found in time $\mathcal{O}(TC)$.
\end{lemma}
\begin{proof}
The sequence is constructed from the end. Therefore, $G_k = G$ and $G_i$ is the graph $G_{i+1}$ without vertex $(c_{i+1},W_{i+1})$. The vertex $(c_i,W_i)$ of $G_i$ is determined in the following way.

Let $G'$ be the graph $G_i$ without edges joining time interval vertices of degree $1$. Since graph $G'$ is also a forest and no time interval vertex is a leaf, the graph $G'$ has a vertex $(c_i,W_i)$ of degree at most 1 and we denote its neighbour by $t_i$ if it exists. In the graph $G_i$, the vertex $(c_i,W_i)$ has at most one neighbour $t_i$ which is not a leaf.

The above construction can be implemented in time which is linear in the size of graph $G$ as follows. The algorithm modifies the graph $G$ so that it contains no vertex corresponding to a time interval of degree at most 1 and it keeps a list of all vertices $(c,W)$ of degree at most 1. An initialization phase can easily ensure these two invariants. Then in every step, some vertex $(c,W)$ is removed from the list and also from the graph, and all neighbours of $(c,W)$ are check to ensure that both invariants are still satisfied.
\end{proof}

\section{Rounding rules} \label{sec:rounding}

This section shows how non-integer values of a vertex $z$ of $P$ are rounded using the order created in the previous section. The rounding rules are summarized in Algorithm \ref{alg:rounding}.

\begin{algorithm}[ht]
\caption{Rounding algorithm. \label{alg:rounding}}
Find a vertex $z$ of $P$ which minimizes a given objective function \eqref{eq:maxpeak}, \eqref{eq:maxabs} or \eqref{eq:maxfluctuation}.\;
Find a sequence $(c_1,W_1),\ldots,(c_k,W_k)$ according to Lemma \ref{lem:order}.\;
$x := z$\;
\For{$i = 1, \ldots, k$}
{
	\tcp{Round values $x_{c_i,t}$ for all $t \in W_i$.}
	\tcp{To simplify the algorithm, assume that $z_{c,0} = 0$ for every $c \in \cal C$.}
	\If{Vertex $t_i$ of Lemma \ref{lem:order} does not exist}{Choice arbitrary $t_i \in W_i$}
	\eIf{$E_{c_i} (\breve{z}_{t_i} - \breve{x}_{t_i}) \ge 0$}
	{
		\For{$t \in W_i$ and $t < t_i$}{$x_{c_i,t} := \floor{\hat{z}_{c_i,t}} - \floor{\hat{z}_{c_i,t-1}}$}
		$x_{c_i,t_i} := 1$\;
		\eIf{$\floor{\hat{z}_{c_i,t_i-1}} = \floor{\hat{z}_{c_i,t_i}}$}
		{\For{$t \in W_i$ and $t > t_i$}{$x_{c_i,t} := \ceil{\hat{z}_{c_i,t}} - \ceil{\hat{z}_{c_i,t-1}}$}}
		{\For{$t \in W_i$ and $t > t_i$}{$x_{c_i,t} := \floor{\hat{z}_{c_i,t}} - \floor{\hat{z}_{c_i,t-1}}$}}
	}{
		\For{$t \in W_i$ and $t < t_i$}{$x_{c_i,t} := \ceil{\hat{z}_{c_i,t}} - \ceil{\hat{z}_{c_i,t-1}}$}
		$x_{c_i,t_i} := 0$\;
		\eIf{$\floor{\hat{z}_{c_i,t_i-1}} = \floor{\hat{z}_{c_i,t_i}}$}
		{\For{$t \in W_i$ and $t > t_i$}{$x_{c_i,t} := \floor{\hat{z}_{c_i,t}} - \floor{\hat{z}_{c_i,t-1}}$}}
		{\For{$t \in W_i$ and $t > t_i$}{$x_{c_i,t} := \ceil{\hat{z}_{c_i,t}} - \ceil{\hat{z}_{c_i,t-1}}$}}
	}
}
\end{algorithm}

The proof of correctness of rounding rules is split into Lemmas \ref{lem:integer}, \ref{lem:feasible} and \ref{lem:approx} proving conditions \eqref{eq:converter}, \eqref{eq:bounds} and \eqref{eq:approx}, respectively.

\begin{lemma}\label{lem:integer}
The solution $x$ found by Algorithm \ref{alg:rounding} is binary.
\end{lemma}
\begin{proof}
Observe that in one iteration of the main for loop in Algorithm \ref{alg:rounding} there is exactly one assignment to a variable $x_{c_i,t}$ for every $t \in W_i$ and the assigned value is 0 or 1. Since for every non-integer value $z_{c,t}$ there exists exactly one equivalence class $W \in T_c$ containing $t$, every non-integral value $z_{c,t}$ has assigned value 0 or 1 exactly once.
\end{proof}

\begin{lemma}\label{lem:feasible}
The solution $x$ found by Algorithm \ref{alg:rounding} satisfies $\floor{\hat{z}_{c,t}} \le \hat{x}_{c,t} \le \ceil{\hat{z}_{c,t}}$ for every $c \in {\cal C}$ and $t \in \cal T$.
\end{lemma}
\begin{proof}
When Algorithm \ref{alg:rounding} starts, the condition \eqref{eq:bounds} holds since $x = z$ and we prove by induction on $i$ that \eqref{eq:bounds} holds after every iteration. Let $t^f$ and $t^l$ be the first and the last time intervals of $W_i$, respectively. Note that $\hat{z}_{c_i,t^f-1} \in \Z$ (assuming that $\hat{z}_{c_i,0} = 0$). Furthermore, $\hat{z}_{c_i,t^l} \in \Z$ unless $\hat{z}_{c_i,T} \notin \Z$ and $t^l$ is the last time interval $t$ with non-integer value in $z_{c_i,t}$.

We restrict our attention on the case $E_{c_i} (\breve{z}_{t_i} - \breve{x}_{t_i}) \ge 0$ since the opposite case is similar. Since $\hat{z}_{c_i,t^f-1} \in \Z$ it follows that $\hat{z}_{c_i,t^f-1} = \hat{x}_{c_i,t^f-1}$. For $t \in T(t^f, t_i)$, we can easily observe by induction on $t$ that $\hat{x}_{c_i,t} = \floor{\hat{z}_{c_i,t}}$. After setting $x_{c_i,t_i} := 1$, it holds that $\hat{x}_{c_i,t_i} = \ceil{\hat{z}_{c_i,t_i}} = \ceil{\hat{z}_{c_i,t_i-1}}$ or $\hat{x}_{c_i,t_i} = \floor{\hat{z}_{c_i,t_i}} = \floor{\hat{z}_{c_i,t_i-1}}+1$. If $\floor{\hat{z}_{c_i,t_i-1}} = \floor{\hat{z}_{c_i,t_i}}$, then for $t \in T(t_i, t^l)$ it also holds that $\hat{x}_{c_i,t} = \ceil{\hat{z}_{c_i,t}}$; otherwise, for $t = T(t_i, t^l)$ it holds that $\hat{x}_{c_i,t} = \floor{\hat{z}_{c_i,t}}$. If $\hat{z}_{c_i,t^l} \in \Z$, then $\hat{z}_{c_i,t^l} = \hat{x}_{c_i,t^l}$ and value $\hat{x}_{c_i,t}$ is unchanged for $t \ge t^l$ in this iteration; otherwise, $\hat{x}_{c_i,t} \in \set{\floor{\hat{z}_{c_i,t}},\ceil{\hat{z}_{c_i,t}}}$ for $t \ge t^l$.
\end{proof}

\begin{lemma}\label{lem:approx}
The solution $x$ found by Algorithm \ref{alg:rounding} satisfies $\left|\breve{z}_t - \breve{x}_t\right| \le E$ for every $t \in \cal T$.
\end{lemma}
\begin{proof}
We prove by induction on $i$ that the lemma holds after every iteration $i$ and moreover, for every vertex of $G_i$ of degree 0 corresponding to time interval $t$ it holds $\breve{x}_t = \breve{z}_t$. The base of the induction is satisfied since graph $G_0$ has no edge and $x = z$.

In the beginning on an iteration $i$, there is at most one $t_i \in W_i$ such that $\breve{x}_t \neq \breve{z}_t$ and if such $t_i$ exists, then $t_i$ is the only non-leaf neighbour in $G_i$ as stated by Lemma \ref{lem:order}. If $E_{c_i} (\breve{z}_{t_i} - \breve{x}_{t_i}) \ge 0$, then setting $x_{c_i,t_i} := 1$ increases the value $\breve{x}_{t_i}$ by at most $E$; otherwise, setting $x_{c_i,t_i} := 0$ decreases the value $\breve{x}_{t_i}$ by most $E$, so the condition \eqref{eq:approx} remains satisfied in both cases. For $t \in W_i \setminus \set{t_i}$, rounding $x_{c_i,t}$ to 0 or 1 changes the value $\breve{x}_{t_i}$ by most $E$, so the condition \eqref{eq:approx} remains satisfied.
\end{proof}

\begin{proof}[Proof of Theorem \ref{thm:round}]
For a given $y$ Lemma \ref{lem:z} gives us $z \in P$ such that $G(z)$ is a forest and Algorithm \ref{alg:rounding} gives us a solution $x$. By Lemma \ref{lem:integer}, the solution $x$ is binary. Lemmas \ref{lem:z} and \ref{lem:feasible} imply that $\floor{\hat{y}_{c,t}} \le \floor{\hat{z}_{c,t}} \le \hat{x}_{c,t} \le \ceil{\hat{z}_{c,t}} \le \ceil{\hat{y}_{c,t}}$ for every $c \in {\cal C}$ and $t \in \cal T$. Lemmas \ref{lem:z} and \ref{lem:approx} imply that $\left|\breve{y}_t - \breve{x}_t\right| = \left|\breve{z}_t - \breve{x}_t\right| \le E$ for every $t \in \cal T$.

By Lemma \ref{lem:z}, $z$ can be found in time $\mathcal{O}(T^2C^2)$. Lemma \ref{lem:order} gives us the order in $\mathcal{O}(TC)$. All rounding in Algorithm \ref{alg:rounding} can be done in $\mathcal{O}(TC)$. Overall, we can find the solution $x$ in time $\mathcal{O}(T^2C^2)$.
\end{proof}

\section{Conclusion} \label{sec:conclusion}

This paper presents polynomial-time approximation algorithms for four variants of peak shaving problems in a model of scheduling a group of converters. For the basic, the maximal and the absolute peak shaving problems, the absolute approximation error is at most $E$; and for the fluctuation peak shaving problem, the error is at most $2E$, where $E$ is the maximal electricity consumption of a converter.

This paper mainly studies the absolute error between an optimal and an approximated solutions instead of the relative error usually used in literature. The main reason for this choice is the fact that the optimal value of our objective functions may be zero, so the relative error is undefined. Furthermore, it is \NP-complete to determine whether there exists a solution with objective value equal to zero in all studied problems except one: The basic peak shaving with positive values $E_c$ for every $c \in \cal C$. In this problem, our algorithm guarantees the relative error to be at most 2.

It is a question whether in practical scenarios it is better to consider the relative or the absolute error. In case studies consisting of many households with heatpumps of similar power, our results guarantees that the difference between an approximated and an optimal solution is bounded when the number of households is increasing. On the other hand, our approximation error may not be satisfactory if e.g. a large power source is included. In this case, it may be possible to adopt our approach as follows. Our method starts by solving a relaxed linear programming problem and here we relax all variables $x_{c,t}$ except the ones corresponding to the large power source. This linear programming problem has $T$ integer variables and $(C-1)T$ continuous variables, so it should be solvable by current solvers. Then, we apply the rounding algorithm only for control variables corresponding to household heatpumps. From Theorem \ref{thm:round} it follows that the absolute error is at most the maximal electricity consumption of all household heatpumps (not the large power source).

\bibliographystyle{plain}
\bibliography{sg}

\newpage
\section{Appendix}

This appendix contains a table of the important symbols used in this paper. Note that a variable $x_{c,t}$ denotes the control of a converter $c \in \cal C$ in time interval $t \in \cal T$ and a solution $x$ means control of all converter during whole planning horizon (similarly for $y$ and $z$). For simplicity, terms like ``by a converter $c$ in time interval $t$'' are omitted in explanations of symbols like $x_{c,t}$.
\medskip

\begin{tabular}[ht]{cl}
$\mathcal{C}$ & set of heating systems of size $C$ \\
$\mathcal{T}$ & set of time intervals of size $C$ \\
$\Z$ & set of integer numbers \\
$A_{c,t}$ & precomputed lower bound on $\hat{x}_{c,t}$ \\
$B_{c,t}$ & precomputed upper bound on $\hat{x}_{c,t}$ \\
$D_{c,t}$ & heat demand from the heating system \\
$E_c$ & electricity consumed by a running converter $c$ \\
$E$ & $= \max_{c \in \cal C} |E_c|$ \\
$F_t$ & base electricity load \\
$G$ & bipartite graph which connects equivalence classes $S_c$ and time intervals of $S_c$ \\
$H_c$ & heat produced by a running converter $c$ \\
$L_{c,t}$ & lower bound on the state of charge of buffer \\
$P$ & polytope $\set{x : \text{ \eqref{eq:sum} and \eqref{eq:relaxed} hold, and $\breve{x}_t = \breve{y}_t$ holds for every $t \in \cal T$}}$ \\
$S_c$ & set of equivalence classes on $T_c$ by relation $\sim_c$ \\
$T$ & number of time intervals \\
$T(t_1,t_2)$ & $= \set{t_1, t_1+1, \ldots, t_2-1}$ if $t_1 \le t_2$ and $\set{t_2, t_2+1, \ldots, t_1-1}$ otherwise \\
$T_c$ & $= \set{t \in {\cal T}: y_{c,t} \notin \Z}$ \\
$U_{c,t}$ & upper bound on the state of charge of buffer \\
$W,W_i$ & one equivalence class of $S_c$ for some converter $c$ \\
$c,c_i$ & indexes of a converter \\
$i,j$ & index of locally defined meaning \\
$m$ & objective function; bound on $\breve{x}_t$ \\
$m^O$ & optimal value of objective function \\
$m^A$ & value of objective function of the approximated solution found by our algorithm \\
$m_l$,$m_u$ & lower and upper bounds on $\breve{x}_t$ in the fluctuation peak shaving problem \\
$s_{c,t}$ & state of charge of buffer $c$ in the beginning of time interval $t$ \\
$t,t_i$ & indexes of a time interval \\
$x_{c,t}$ & operational state of the converter \\
$\hat{x}_{c,t}$ & $= \sum_{i=1}^t x_{c,i}$; similarly for $\hat{y}_{c,t}$ and $\hat{z}_{c,t}$ \\
$\breve{x}_t$ & $= F_t + \sum_{c \in \cal C} E_c x_{c,t}$; similarly for $\breve{y}_t$ and $\breve{z}_t$ \\
$y$ & optimal solution of the relaxed problem \\
$z(\alpha)$ & point of a line segment parametrized by $\alpha$ belonging into $P$ for $-\epsilon \le \alpha \le \epsilon$ \\
$z$ & a vertex of $P$ \\
$\sim_c$ & relation on $T_c$ such that $t_1 \sim_c t_2$ if and only if $\hat{y}_{c,t} \notin \Z$ for every $t = T(t_1, t_2)$ \\
$(c,W)$ & vertex of the graph $G$; $W$ is equivalence class of $S_c$ for converter $c$ \\
$\floor{a}$ & largest integer value not greater than argument \\
$\ceil{a}$ & smallest integer value not smaller than argument \\
\end{tabular}

\end{document}